\documentclass[conference,10pt]{IEEEtran}

\usepackage{cite}

\ifCLASSINFOpdf
  \usepackage[pdftex]{graphicx}
   \graphicspath{{.}{./figs/}}
   \DeclareGraphicsExtensions{.pdf,.jpeg,.png,.pdf}
\else

\fi
\usepackage[cmex10]{amsmath}
\usepackage{algorithmic}

\usepackage{array}

\usepackage{mdwmath}
\usepackage{mdwtab}

\usepackage[caption=false,font=footnotesize]{subfig}
\usepackage{fixltx2e}

\usepackage{url}

\newcommand{\bl}[1]{{\color{black}{#1}}}

\newcommand{\B}{\tilde{\A}}
\newcommand{\one}{\boldsymbol{1}}

\newcommand{\m}{{N}}

\newcommand{\pow}[1]{\Gamma_{#1}^\ell}

\newcommand{\nc}{\left[\n\right]}

\usepackage{algorithm,amsmath,amssymb}
\usepackage{algorithmic}

\newcommand{\Ex}{\mathbb{E}}

\newcommand{\mm}{{N}}

\renewcommand{\m}{{\mm}}

\newcommand{\nor}{\frac{1}{\sqrt{\mm}}}
\newcommand{\A}{{\Phi}}

\newtheorem{proposition}{Proposition}

\newtheorem{theorem}{Theorem}
\newtheorem{lemma}{Lemma}
\newtheorem{remark}{Remark}

\newcommand{\has}{{{\hat{\as}}}}

\newcommand{\n}{{\cal C}}

\newcommand{\as}{{\alpha}}
\newcommand{\yy}{{f}}
\newcommand{\icd}{}

\newcommand{\ignore}[1]{}           

\usepackage{color}
\usepackage{algorithmic}
\renewcommand{\bl}[1]{{{\color{black} {#1}}}}
\newcommand{\alert}[1]{{{\color{black} {#1}}}}
\begin{document}
%
\title{Sparse Reconstruction via The Reed-Muller Sieve}

\author{\IEEEauthorblockN{Robert Calderbank}
\IEEEauthorblockA{Department of Electrical Engineering\\ 
Princeton University\\
Princeton, NJ 08540\\
Email: calderbk@math.princeton.edu}
\and
\IEEEauthorblockN{Stephen Howard}
\IEEEauthorblockA{DSTO\\
PO Box 1500\\
Edinburgh 5111, Australia\\
Email: show@ee.unimelb.edu.au}
\and
\IEEEauthorblockN{Sina Jafarpour}
\IEEEauthorblockA{Department of Computer Science\\
Princeton University\\ Princeton NJ 08540\\
Email: sina@cs.princeton.edu}}

\maketitle

\begin{abstract}
\boldmath
\bl{This paper introduces the Reed Muller Sieve, a deterministic measurement matrix for compressed sensing. The columns of this matrix are obtained by exponentiating codewords in the quaternary second order Reed Muller code of length $N$. For $k=O(N)$, the Reed Muller Sieve improves upon prior methods for identifying the support of a $k$-sparse vector by removing the requirement that the signal entries be independent. The Sieve also enables local detection; an algorithm is presented with complexity $N^2 \log N$ that detects the presence or absence of a signal at any given position in the data domain without explicitly reconstructing the entire signal. Reconstruction is shown to be resilient to noise in both the measurement and data domains; the $\ell_2 / \ell_2$ error bounds derived in this paper are tighter than the $\ell_2 / \ell_1$ bounds arising from random ensembles and the $\ell_1 /\ell_1$ bounds arising from expander-based ensembles.}
\end{abstract}
\IEEEpeerreviewmaketitle

\section{Introduction}
The central goal of compressed sensing is to capture attributes of a signal using very few measurements. In most work to date, this broader objective is exemplified by the important special case in which the measurement data constitute a vector $f=\A\as+e$, where $\A$ is an $\mm\times\n$ matrix called the \textit{sensing matrix}, $\as$ is a vector in $\mathbb{C}^\n$, which can be well-approximated by a $k$-sparse vector, where a $k$-sparse vector is a vector which has at most $k$ non-zero entries, and $e$ is additive measurement noise. 

The role of random measurement in compressive sensing (see \cite{CRT1} and \cite{Donoho}) can be viewed as analogous to the role of random coding in Shannon theory. Both provide worst-case performance guarantees in the context of an adversarial signal/error model. In the standard paradigm, the measurement matrix is required to act as a near isometry on all $k$-sparse signals (this is the Restricted Isometry Property or RIP introduced in \cite{CT}). Basis Pursuit \cite{CRT1,CRT2} or Matching Pursuit algorithms \cite{greed,NT} can then be used to recover any $k$-sparse signal from the $N$ measurements. These algorithms rely heavily on matrix-vector multiplication and their complexity is super-linear with respect to $\n$, the dimension of the data domain. The worst case complexity of the convex programs Basis Pursuit \cite{CRT1}, LASSO \cite{CP07} and the Dantzig Selector \cite{dantzig} is $\n^3$ though the average case complexity is less forbidding. Although it is known that certain probabilistic processes generate $\mm\times\n$ measurement matrices that satisfy the RIP with high probability, there is no practical algorithm for verifying whether a given measurement matrix has this property. Storing the entries of a random sensing matrix may also require significant resources.

The Reed Muller Sieve is a deterministic sensing matrix. The columns are obtained by exponentiating codewords in the quaternary second order Reed Muller code; they are uniformly and very precisely distributed over the surface of an $\mm$-dimensional sphere. Coherence between columns reduces to properties of these algebraic codes and we use these properties to show that recovery of $k$-sparse signals is possible with high probability.   

When the sparsity level $k=O\left(\sqrt{N}\right)$, recovery is possible using the algorithm presented in \cite{strip} and the reconstruction complexity is only  $k\,\mm \log^2\mm$. The prospect of designing matrices for which very fast recovery algorithms are possible is one of the attractions to deterministic compressive sensing. When the sparsity level $k=O(N)$ recovery is possible using the algorithm described in this paper. Reconstruction complexity is $\mm\,\n$, the same as for both CoSaMP \cite{NT} and SSMP \cite{DM}.

We note that there are many important applications where the objective is to identify the signal model (the support of the signal $\as$). These include network anomaly detection where the objective is to characterize anomalous flows and cognitive radio where the objective is to characterize spectral occupancy. The Reed Muller sieve improves on results obtained by Cand\`es and Plan \cite{CP07} in that for $k=O(\mm)$ it is able to identify the signal model without requiring that the signal entries $\as_i$ be independent.

Reconstruction of a signal from sensor data is often not the ultimate goal and it is of considerable interest in imaging to be able to deduce attributes of the signal from the measurements without explicitly reconstructing the full signal. We show that the Reed Muller Sieve is able to detect the presence or absence of a signal at any given position in the data domain without needing to first reconstruct the entire signal. The complexity of such detection is $\mm^2 \log\mm$. This makes it possible to quickly calculate thumbnail images and to zoom in on areas of interest.

There are two models for evaluating noise resilience in compressive sensing. We provide an average case error analysis for both the stochastic model where noise in the data and measurement domains is usually taken to be iid white Gaussian, and the deterministic model where the goal is to approximate a compressible signal. It is the geometry of the sieve, more precisely the careful design of coherence between columns of the measurement matrix, which provides resilience to noise in both the measurement and the data domain. Our analysis points to the importance of both the average and the worst-case coherence. 
 
We show that the $\ell_2$ error in reconstruction is bounded above by the $\ell_2$ error of the best $k$-term approximation. This type of $\ell_2 / \ell_2$ bound is tighter than the $\ell_2 /\ell_1$ bounds arising from random ensembles \cite{CRT1,NT} and the $\ell_1/\ell_1$ bounds arising from expander-based ensembles \cite{IR,sina}. We emphasize that our error bound is for average-case analysis and note that results obtained by Cohen et. al. \cite{best} show that worst-case $\ell_2/\ell_2$ approximation is not achievable unless $\mm=O(\n)$.

\section{Two Fundamental Measures of Coherence}
\label{sec:pre}
{

Throughout this paper we also abbreviate $\{1,\cdots,\n\}$ by $\nc$. We shall use the notation $\varphi_j$ for the $j^{th}$ column of the sensing matrix; its entries will be denoted by $\varphi_j(x)$, with the row label $x$ varying from $0$ to $\mm-1$. We consider sensing matrices for which reconstruction of $\as$ is guaranteed in expectation only, and so we need to be precise about our signal model. 

A signal $\as\in\mathbb{R}^\n$ is $k$-sparse if it has at most $k$ non-zero entries. The support of the vector $\as$, denoted by $\mbox{Supp}(\as)$, contains the indices of the non-zero entries of $\as$. Let $\pi\nobreak=\nobreak\{\pi_1,\cdots,\pi_\n\}$ be a uniformly random permutation of $\nc$. Since our focus is on the average case analysis, we always assume that $\as$ is a $k$-sparse signal with $\mbox{Supp}(\as)=\{\pi_1,\cdots,\pi_k\}$ and the values of the $k$ non-zero entries of $\as$ are specified by $k$ \textit{fixed} numbers $\as_1,\cdots,\as_k$. We shall also define $|\as_{\min}|\doteq \min_{{i\,:\as_i\neq 0}} |\as_i|.$ 

The following proposition is proved by Calderbank et. al \cite{strip,tech1} and plays a key role in our analysis
\begin{proposition}\label{foundation} Let $\alpha$ be a $k$-sparse vector with support $S=\{\pi_1,\cdots,\pi_k\}$. Let $h$ be a function from $\nc\times\nc$ to $ \mathbb{R}$, and let $\A$ be an $\mm \times \n$ sensing matrix. If the following two conditions hold:
\begin{itemize}
\item(St1). $\mu\doteq\max_{i\neq j}\left|h(i,j) \right|\leq \mm^{-\eta}~~\mbox{for}\,\,0\leq\eta\leq 0.5.$
\item(St2).  $\nu\doteq \max_i\,\frac{1}{\n-1}\left| \sum_{j\neq i}h(i,j) \right|\leq \mm^{-\gamma}~~\mbox{for}\,\,\gamma\geq1.$
\end{itemize}
Then for all positive $\epsilon$ and for all $k$ less than $\min\left\{{\frac{\n}{2}},\frac{\epsilon \mm^{\gamma}}{4}\right\}$, with probability $1-k\,\n\, \exp\left\{-\frac{\mm^{2\eta}\epsilon^2}{32} \right\}$ the following three statements hold:
\begin{itemize}
\item(Sp1) For every $w$ in $\nc-S$:\, $\left|\sum_j \alpha_j h(w,\pi_j)\right|\leq \epsilon \|\as\|_2$.
\item(Sp2) For every index $i$ in $\{1,\cdots,k\}$: $$\left|\sum_{j\neq i} \alpha_i h(\pi_i,\pi_j)\right| \leq \epsilon \|\as\|_2.$$
\item (Sp3) If $\n{|\as_{\min}|^2}\geq{\|\as\|^2}$, then 
$$\left|\sum_i \sum_{j\neq i} \alpha_i \overline{\alpha_j} h(\pi_i,\pi_j)\right| \leq \epsilon \|\as\|^2.$$
\end{itemize}
\end{proposition}
\begin{remark}
A matrix satisfying conditions (St1) and (St2) is called a StRIP-able matrix. Similarly a matrix satisfying conditions (Sp1), (Sp2), and (Sp3) is called a StRIP matrix. Proposition~\ref{foundation} states that StRIP-ability is a sufficient condition for the StRIP property.
\end{remark}

\section{The Reed-Muller Sieve}
\label{sec:matrices}

Let m be an odd integer. The measurement matrix $\B\nobreak=\nobreak\B_{m,r}$ has $2^m$ rows indexed by binary $m$-tuples $x$ and $2^{(r+1)m}$ columns indexed by $m\times m$ binary symmetric matrices $Q$ in the Delsarte-Goethals set $DG(m,r)$. The entry $\varphi_{Q}(x)$ is given by $
\varphi_{Q}(x)=\imath^{xQx^\top},
$ and all arithmetic in the exp{\icd{ressions}} $xQx^\top$ takes place in the ring of integers modulo $4$. The matrices in $DG(m,r)$ form an $(r+1)$-dimensional binary vector space and the rank of any non-zero matrix is at least $m-2r$ (see \cite{DG} and also \cite{tech1} for an alternative description). The Delsarte-Goethals sets are nested
$$DG(m,0) \subset DG(m,1)\subset\cdots \subset DG\left(m,\frac{m-1}{2}\right).$$
The set DG(m,0) is called the Kerdock set and it contains $2^m$ nonsingular matrices with distinct main diagonals. The vector of length $2^m$ with entries $xQx^\top$ is a codeword in the quaternary Delsarte-Goethals code \cite{H}.

In Section~\ref{sec:witness} we will apply the following result on partial column sums to guarantee  fidelity of reconstruction.
\alert{
\begin{proposition}\label{new}
Let $V$ and $W$ be two binary symmetric matrices and let ${\cal N}_W$ and ${\cal N}_{V-W}$ be the null spaces of $W$ and $V-W$. If ${\rm S}=\sum_{x,a }\imath^{aVa^\top+xWx^\top+2aWx^\top},$
then $\left|{\rm S}^2\right|=2^{2m}\,\, 2^{|{\cal N}_W|+|{\cal N}_{V-W} |}$.
\end{proposition}
\begin{IEEEproof} We have
$${\rm S}^2=\sum_{\substack{a,b,x,y}} \imath^{aVa^\top+bVb^\top+xWx^\top+yWy^\top+2aWx^\top+2bWy^\top} .$$
Changing variables to $z=x+y,\,c=a+b,\,y$ and $b$ yields
\begin{eqnarray}\label{vanish} {\rm  S}^2&=&\sum_{c,z} \imath^{cVc^\top+zWz^\top+2cWz^\top} \\ \nonumber &&\left(\sum_b (-1)^{\left(cV+d_V+zW\right)b^\top} \right) \left(\sum_y (-1)^{\left(zW+d_W+cW\right)y^\top} \right)  .\end{eqnarray}
The terms in Equation~\eqref{vanish} vanishes unless $cV+d_V+zW=0$ and $zW+d_W+cW=0$ simultaneously. Hence, we can rewrite Equation~\eqref{vanish} as
$$2^{2m} \sum_{\substack{c,z\\(c+z)W=d_W\\c(V+W)=d_V+d_W}} \imath^{(c+z)W(c+z)^\top+c(V-W)c^\top}.$$
Write $c=c_1+e$ with $c_1(V+W)=d_V+d_W$ and $e(V+W)\nobreak=\nobreak0$, and $c+z=(c_2+z_2)+f$ with $(c_2+z_2)W=d_W$ and $f W=0$. Then
$$\left|{\rm S} \right|^2=2^m \left|\sum_f \imath^{fWf^\top} \right| \left|\sum_e \imath^{e(V-W)e^\top} \right|=2^{2m}\,2^{|{\cal N}_W|+|{\cal N}_{V-W}|}.$$
\end{IEEEproof}
}

Proposition~\ref{new} bounds the worst case coherence between columns of $\B_{m,r}$. We bound average coherence by dividing the columns into a set H indexed by the matrices in $DG(m,r)$ with zero diagonal, and a set D indexed by the matrices in the Kerdock set. The columns in H form a group under pointwise multiplication.

\begin{lemma}\label{average} Let $\B$ be a $DG(m,r)$ sensing matrix. Then
$$\sum_{j\neq i} \varphi_i^\dag\varphi_j=\frac{-\mm}{\n-1}\mbox{ for every index $i$}.$$\end{lemma}
\begin{proof}
Any column can be written as a pointwise product $hd$ with $h$ in H and $d$ in D. Average coherence with respect to $hd$ is then
\begin{eqnarray} (\n-1)^{-1} \sum\limits_{(h',d')\neq(h,d)}  d^{-1} h^{-1}h' d'\label{goodgroup}.
\end{eqnarray}
If $d\neq d'$, then $h'$ ranges over all elements of H and $\sum_{h'} h^{-1}h'=0$. Otherwise $h' \neq h$ and $\sum_{h' \neq h} h^{-1}h'=-\one$. In this case $d^{-1}\one d =\mm$, which completes the proof.
\end{proof}
The normalized Delsarte-Goethals sensing matrix is given by $\A=\nor\B$, and we have now proved 
\begin{theorem}\label{proposed}
 The normalized matrix $\A$ satisfies Condition (St1) with $\eta = 1+r$ and Condition (St2) with $\gamma = \frac{1}{2} \left(1-\frac{2r}{m}\right).$
\end{theorem}
\subsection{Noise Shaping}
The tight-frame property of the sensing matrices makes it possible to achieve resilience to noise in both the data and measurement domains. Note that the factor $\frac{\n}{\mm}$ that appears in Lemma~\ref{noise:bound} can be reduced by subsampling the columns of $\A$.
\begin{lemma}\label{noise:bound}
Let $\varsigma$ be a vector with $\n$ iid ${\cal N}(0,\sigma_d^2)$ entries and $e$ be a vector with $\m$ iid ${\cal N}(0,\sigma_m^2)$ entries. Let $\hbar=\A\varsigma$ and $u=\hbar+e$. Then $u$ contains $\mm$ entries, sampled iid from ${\cal N}\left(0, \sigma^2\right)$, where $\sigma^2=\frac{\alert{}\n}{\mm}\sigma_d^2+\sigma_m^2$ and with probability $1-\frac{1}{\n}$, $\|u\|\leq \sqrt{\mm\log \n}\, \sigma$.
\end{lemma}
\begin{proof}
\alert{Each element of $\hbar$ is an independent Gaussian random variable with zero mean and variance at most $\frac{\n}{\mm}$. Hence, each element of $u$ is a Gaussian random variable with zero mean and variance at most $\frac{\n}{\mm}\sigma_d^2+\sigma_m^2$. 
It therefore follows from the tail bound on the maximum of $\mm$ arbitrary complex Gaussian random variables with bounded variances that 
$$\Pr\left[\|u\|_\infty \geq \sqrt{2\sigma^2 \log \n} \right] \leq 2\,\left(\sqrt{2\pi\log \n}\,\n \right)^{-1}\leq \n^{-1}.$$
}
\end{proof}
\section{The Chirp Reconstruction Algorithm}
\label{sec:chirp}

\begin{algorithm}
\caption{Chirp Reconstruction}
     \label{alg1}
   \begin{algorithmic}[1]
        \FOR{$i=1,\cdots \mm$}
                 \STATE Choose the next binary offset $a$.
                 \STATE Pointwise multiply $f$ with a shifted version of itself.
                 \STATE Compute the fast Hadamard transform:  $\Gamma_{a}^\ell\left(f\right)$.
    	\STATE \label{newstep}For each $\Delta\in\nc$, calculate $\Lambda_{\Delta,a}=\imath^{-aQ_{\Delta}a^\top} \Gamma_a^{a\,Q_\Delta}(f)$.
		 \ENDFOR
		 \STATE  \label{newstep2}For each $\Delta\in\nc$, take the average of $\Lambda_\Delta$ over all $\Lambda_{\Delta,a}$.
		 \STATE \label{finder}Let $S$ be the position of the $k$ highest (in magnitude) average peaks.
             \STATE \label{strip:step}Output $\hat{\as}= (\A_S^\dag \A_S)^{-1} \A_S^\dag \yy$. 
   \end{algorithmic}
   \end{algorithm}
Chirp reconstruction identifies the signal model (the support of the $k$ significant entries) by analyzing the power spectrum of the pointwise product of the superposition $f$ with a shifted version of itself. The Walsh-Hadamard transform of this pointwise product is the superposition of $k$ Walsh functions and a background signal produced by cross-correlations between the $k$ significant entries and cross-correlations between these k entries and noise in the data domain. We shall prove that the energy in this background signal is uniformly distributed across the Walsh-Hadamard bins, and that with overwhelming probability this background bin energy is sufficiently small to enable threshold detection of the $k$ tones. We show that sparse reconstruction is possible for $k=O(\mm)$ by averaging over all possible shifts. Note that the original chirp reconstruction algorithm analyzed in \cite{HSC} has minimal complexity $k\,\mm\,\log^2\mm$ but reconstruction is only guaranteed for $k=O\left(\sqrt{\mm}\right)$. Our main result is the following theorem.  
\begin{theorem}
\label{thr:chirp1}
Let $\A$ be an $\mm\times \n$ normalized $DG(m,r)$ matrix. Let $\as$ be a $k$-sparse vector with uniformly random support contaminated by Gaussian white noise with variance $\sigma_d^2$ . Let $f=\A\as+e$, where the measurement errors are white Gaussian with variance $\sigma_m^2$. Then if $k \leq \frac{\n}{2}$ and \begin{equation}\label{c1}\mm^{\alert{1-\frac{2r}{m} } }\geq \frac{36\sqrt{\log \n}\|\as\|^2}{|\as_{\min}|^2},\end{equation} and \begin{equation}\label{c2}\sigma_m^2+\frac{\n}{\mm}\sigma_d^2\leq \left(\frac{|\as_{\min}|^2 \mm^{\alert{\frac{1}{2}-\frac{2r}{m}} }}{{36\log \n}\|\as\|}\right)^2,\end{equation}  then with probability \bl{$1-\frac{3}{\n}$}, chirp reconstruction recovers the support of $\as$, and furthermore,\begin{equation}\label{c:constant}\|\A\left(\as-\has\right)\|^2\leq c k \log \n \left( \frac{\n}{\mm}\sigma_d^2+\sigma_m^2\right),\end{equation} where $c$ is a constant. 
\end{theorem}
\vskip0.1cm
The fast Hadamard transform is used to calculate the power spectrum across all $\mm$ Hadamard bins. Each bin $\ell$ has the value 
\begin{equation}\label{power:spectrum}\Gamma_a^\ell(f)\doteq \nor \sum_{x=1}^\mm (-1)^{\ell x^\top} f(x+a)\overline{f(x)} \end{equation}
Given the offset $a$, evidence for the presence or absence of a signal at position delta in the data domain resides in the Hadamard bin $\ell=a\, Q_{\Delta}$. After aligning the phase, the final step is to average over all offsets $a$. The notation $\mathbb{E}_a$ emphasizes that the average is taken over all offsets. The following theorem shows that $\Ex_a\left[\Lambda_\Delta(y)\right]$ consists of $k$ distinct Walsh tones staying on top of a uniform chirp-like residual term.

\begin{theorem} \label{support:rec} Let $u=\A\varsigma+e$ denote the overall noise. Then as long as $k\leq \frac{\n}{2}$, with probability $1-\frac{1}{\n}$, for every index $\Delta$ in $\nc$
$$\Ex_a\left[\Lambda_\Delta(f)\right]=\sum_{i=1}^k \frac{|\as_i|^2}{\sqrt{\mm}} \delta_{\Delta,\pi_i}+{ R}^\Delta(f), $$
\bl{where $R^{\Delta}(f)$ consists of the chirp-like and signal/noise cross correlation terms, and}
\bl{\begin{equation}\left|{R}^\Delta(f)  \right|\leq \frac{9\sqrt{\log \n}\|\as\|^2}{\mm^{\alert{ \frac{3}{2}-\frac{2r}{m} } }} + \frac{9\sqrt{\log \n}\|\as\|\|u\|}{\mm^{ \alert{  \frac{3}{2}-\frac{2r}{m} } }}.\label{eqnoise:bound}\end{equation}}
\end{theorem}

Theorem~\ref{support:rec} is proved in Section~\ref{sec:witness}. The next lemma provides a lower bound on the number of required measurements.
\begin{lemma}\label{rec} Let $\varsigma$ and $e$ be the white Gaussian data and measurement noise vectors with variances $\sigma_d^2$ and $\sigma_m^2$ respectively. Let $\sigma^2=\sigma_m^2+\frac{\n}{\mm}\sigma_d^2$. If $k\leq \frac{\n}{2}$, $$\mm^{\alert{ 1-\frac{2r}{m} } }\geq \frac{36\sqrt{\log \n}\|\as\|^2}{|\as_{\min}|^2}\mbox{ and } \sigma\leq \frac{|\as_{\min}|^2 \mm^{\alert{ \frac{1}{2}-\frac{2r}{m}  }  }}{{36\log \n}\|\as\|},$$ then with probability $1-\frac{2}{C}$, chirp reconstruction successfully recovers the positions of the $k$ significant entries of $\as$.
\end{lemma}
\begin{proof} 
Chirp detection generates $k$ Walsh tones with magnitudes at least $\frac{|\as_{\min}|^2}{\sqrt{\mm}}$ above a uniform background signal. Furthermore, with probability at least \bl{$1-\frac{1}{\n}$} every background signal at every index is bounded by the right hand side of~\eqref{eqnoise:bound}. Hence, if the right hand of~\eqref{eqnoise:bound} is smaller than $\frac{|\as_{\min}|^2}{2\sqrt{\mm}}$ then the $k$ tones pop up and we can detect them by thresholding. Hence, we need to ensure that $$\frac{9\|\as\|^2\sqrt{\log\n}}{\mm^{\alert{ \frac{3}{2}-\frac{2r}{m} } } } \leq \frac{|\as_{\min}|^2}{4\sqrt{\mm}}\mbox{ and }\frac{9\sqrt{\log \n}\|\as\|\|u\|}{\mm^{\alert{\frac{3}{2}-\frac{2r}{m} } }}  \leq \frac{|\as_{\min}|^2}{4\sqrt{\mm}}.$$ Now Lemma~\ref{noise:bound} states that with probability $1-\frac{1}{\n}$, $\|u\|\nobreak\leq\nobreak \sqrt{\mm\log \n}\, \sigma$. Consequently, to provide successful support recovery we need to assure that  $$\frac{9 \sqrt{\mm}\,\log \n \|\as\|\sigma}{\mm^{\alert{ \frac{3}{2}-\frac{2r}{m} } }}  \leq \frac{|\as_{\min}|^2}{4\sqrt{\mm}}.$$ 
 \end{proof}

\begin{proof}[Proof of Theorem~\ref{thr:chirp1}]  Lemma~\ref{rec} guarantees that with probability $1-\frac{2}{\n}$ Chirp Detection successfully recovers the support $S$ of $\as$. We then approximate the values of $\as$ by regressing $f$ onto $S$. By Lemma~\ref{noise:bound}, without loss of generality we can assume that $\as$ is exactly $k$-sparse and the measurement errors are white Gaussian with variance $\sigma^2$. We have 
$\|\A\as-\A\has\|^2\nobreak \leq\nobreak\| {\cal P}_{S} u\|^2$, where ${\cal P}_{S} u$ denotes the projection of the noise vector onto the space spanned by $\Phi_S$. Now it follows from the Gaussian tail bound (See \cite{omid}), that with probability $1-\frac{1}{\n}$, $\|{\cal P}_{S} u\|^2 \leq c k\,\log \n \, \sigma^2$ (where $c$ is a constant). Therefore $$\|\A\as-\A\has\|^2\leq c\,k\,\log\n\,\left(\sigma_m^2+\frac{\n}{\mm}\sigma_d^2\right).$$
\end{proof}
\begin{remark}
We have focused on stochastic noise but we have derived similar results for the best $k$-term approximation $\alpha_{1\rightarrow k}$ in the context of the deterministic noise model (see \cite{tech1}). By combining the Markov inequality with Theorem~\ref{proposed} we have shown that for every $\delta'$ 

$$\Pr_\pi\left[\|u\|_2 \geq \|e\|_2+\frac{1}{\sqrt{\delta'}}\|\as-\as_{1\rightarrow k}\|_2\right]\leq \delta'.$$ This $\ell_2/\ell_2$ error bound for average-case analysis stands in contrast to results obtained by Cohen et al \cite{best} showing that worst-case $\ell_2/\ell_2$ approximation is not achievable unless $\mm\nobreak=\nobreak O(\n)$.
\end{remark}
\begin{remark}
There are many important applications where the objective is to identify the signal model (the support of the signal $\as$). Note that in contrast to \cite{CP07} chirp reconstruction does not require that the component signals $\alpha_i$ be independent. When the $\mm\times k$ submatrix is well conditioned the approximation bound \eqref{c:constant} in the measurement domain ($\|\A(\as-\has)\|$) can be translated directly to approximation bound in the data domain ($\|\as-\has\|$).
\end{remark}
\begin{remark}
Chirp reconstruction is able to detect the presence or absence of a signal at any given index $\Delta$ in the data domain without needing to first reconstruct the entire signal. The complexity of detection is $N^2 \log\mm$. If the signal $\as$ were the wavelet decomposition of an image, then chirp reconstruction can be applied to the measured signal to recover thumbnails and to zoom in on areas of interest.
\end{remark}
\section{The Sieving of Evidence}
\label{sec:witness}
We now prove Theorem~\ref{support:rec} in the special case where there is no noise $(u=0)$. We begin by expanding $f$ as $\sum_{i=1}^k \alpha_i \varphi_{\pi_i}$ where $\varphi_{\pi_i}(x)\,=\, \imath^{xQ_{\pi_i}x^\top}$.

  \begin{lemma}\label{fixed}  Let $\A$ be a normalized $\mm \times \n$ sensing matrix. Then for any offset $a$ and Hadamard bin $\ell$ in $\mathbb{F}_2^m$
 \begin{eqnarray}\label{robert}
 \pow{a}(f)= \sum_{i=1}^k \frac{i^{aQ_{\pi_i}a^\top}\left|\as_i\right|^2}{{\mm^{\frac{1}{2}}}} \delta_{aQ_{\pi_i},\ell}+{ R}_a^\ell(f),
 \end{eqnarray}
 where  ${ R}_a^\ell(f)$ is the quantity
 $$
\frac{1}{\mm^{\frac{3}{2}}} \sum_{i=1}^k\sum_{j\neq i} \as_i \overline{\as_j} \imath^{aQ_{\pi_i}a^\top} \sum_x (-1)^{\left(aQ_{\pi_i}+\ell\right)x^\top} i^{x\left(Q_{\pi_i}-Q_{\pi_j}\right)x^\top} .
 $$
 \end{lemma}

Next we analyze the power of chirp reconstruction to detect the presence of a signal at some index $\Delta$.
 \begin{lemma}\label{lem:good}If $\Delta$ is any index in $\n$ then after Step~\ref{newstep} of Chirp Reconstruction, for all $\Delta$ in $\nc$ 
$$ \Ex_a\left[ \Lambda_{\Delta}(f)\right]= \sum_{i=1}^k \frac{|\as_i|^2}{\sqrt{\mm}} \delta_{\Delta,\pi_i}+{ R}^{\Delta}(f),$$
where $R^{\Delta}(f)=\mathbb{E}_a\left[\imath^{-aQ_{\Delta}a^\top} R_a^{aQ_{\Delta}} \right]$, and 
$$\left|{R}^\Delta(f)\right|\leq \bl{9} \frac{\|\as\|^2}{\mm^{\alert{\frac{3}{2}-\frac{2r}{m}} }}.$$
 \end{lemma}
\begin{proof}
If $\Delta = \pi_i$ for some $i$ then the signal $\as_i$ contributes $\imath^{a Q_{\Delta} a^\top} \frac{|\as_i|^2}{\sqrt{\mm}}$ to the Hadamard bin $a\,Q_{\Delta}$. Rotation by $\imath^{-a Q_{\Delta} a^\top}$ and averaging over all offsets $a$ accumulates evidence $\frac{|\alpha_i|^2}{\sqrt{\mm}}$ for the presence of a signal at the index $\Delta$.

If $\Delta\neq \pi_i$ for any $i$, then it is only the cross-terms that contribute to the Hadamard bin  $a\,Q_{\Delta}$. Rotation of the contribution  $R_a^{aQ_{\Delta}}(f)$ by $\imath^{-a Q_{\Delta} a^\top}$ and averaging over all offsets $a$ produces the background signal against which we perform threshold detection. 

Define $h(\pi_i,\pi_j)$ by
$$\Ex_a\left[ i^{a\left(Q_{\pi_i}-Q_\Delta\right)a^\top} \sum_x (-1)^{\left(aQ_{\pi_i}-{aQ_\Delta}\right)x^\top} i^{x\left(Q_{\pi_i}-Q_{\pi_j}\right)x^\top} \right].$$
We will bound $R^{\Delta}(f)$ by applying Proposition~\ref{foundation} so we need to verify Conditions (St1) and (St2). By Proposition~\ref{new}, the \alert{expectation above is always bounded in magnitude by the term $2^{2r}$. Hence}
it remains to bound average coherence, and here we show that $\max_i  \left|\mathbb{E}_{j\neq i} h(i,j) \right| \leq \frac{1}{\n-1}$. We rewrite $h(\pi_i,\pi_j)$ as
$$\frac{1}{\mm} \sum_x \imath^{x\left(Q_{\Delta}-Q_{\pi_j}\right)x^\top}  \sum_a \imath^{(a+x) \left(Q_{\pi_i}-Q_\Delta\right) (a+x)^\top }.  $$
Note that as $a$ ranges over the finite field $\mathbb{F}_2^m$, $a+x$ also ranges over $\mathbb{F}_2^m$. Therefore  $\sum_a \imath^{(a+x) \left(Q_{\pi_i}-Q_\Delta\right) (a+x)^\top }$ is a constant column sum, independent of the choice of $j$, and has magnitude smaller than $\mm$. As a result
$$  \left|\mathbb{E}_{j\neq i} h(i,j) \right| \leq \left|\mathbb{E}_{j\neq i}\left[ \sum_x \imath^{x\left(Q_{\Delta}-Q_{\pi_j}\right)x^\top}   \right] \right|.$$
Lemma~\ref{average} then implies that  $$\left|\mathbb{E}_{j\neq i}\left[ \sum_x \imath^{x\left(Q_{\Delta}-Q_{\pi_j}\right)x^\top}   \right] \right|\leq \frac{1}{\n-1}.$$

We have now shown with respect to $h$ that $\A$ satisfies Condition (St1) with $\eta = 1+r$ and Condition (St2) with $\gamma = \alert{-\frac{2r}{m}}.$ It then follows from applying Proposition~\ref{foundation} with $$\epsilon={\mm}^{\alert{\frac{2r}{m}}} {9\sqrt{\log \n}},$$
that with probability $1-\frac{1}{\n}$,  
$$\frac{1}{\mm^{\frac{3}{2}}} \left|\sum_{i} \sum_{j\neq i} \as_i \overline{\as_j}h(\pi_i,\pi_j)\right|
$$
is bounded by $9{\mm^{\alert{\frac{2r}{m}-\frac{3}{2} } }}{\|\as\|^2\sqrt{\log \n}}.$
\end{proof}
\section{Noise Resilience}
\label{sec:noise}
 When noise is present in the data domain or in the measurements, the power spectrum contains extra terms arising from the signal/noise cross correlation and noise autocorrelation. It is natural to neglect noise autocorrelation and to focus on the cross correlation between signal and noise.
 
 Let $y=\A\as$. At the end of Step~\ref{newstep2}, for each index $\Delta$, the signal/noise cross correlation can be represented as 
 \begin{equation}\label{cross:noise}
\mathbb{E}_a\left[\frac{\imath^{-aQ_\Delta a^\top}}{\sqrt{\mm}} \left(\sum_x  \overline{y(x)}u(x+a)(-1)^{aQ_\Delta x^\top}\right)\right].
 \end{equation}
 We have modified the argument used to prove Lemma~\ref{lem:good} to show that with probability $1-\frac{1}{\n}$,  the signal/noise cross correlation term is uniformly bounded by $9\,\mm^{\alert{\frac{2r}{m}-\frac{3}{2}}}{\bl{}\|\as\|\|u\|}$ (see \cite{tech1} for more details).
 
\section{Conclusion}
\label{sec:conc}
\bl{In compressed sensing the entries of the measurement vector constitute evidence for the presence or absence of a signal at any given location in the data domain. We have shown that the Reed Muller sieve is able to identify the support set without requiring that the signal entries be independent. We have also demonstrated feasibility of local decoding where attributes of the signal are deduced from the measurements without explicitly reconstructing the full signal. Our reconstruction algorithms are resilient to noise and the $\ell_2/ \ell_2$ error bounds are tighter than the $\ell_2 / \ell_1$ bounds arising from random ensembles and the $\ell_1 /\ell_1$ bounds arising from expander-based ensembles.}
\bibliographystyle{IEEEbib} 
\bibliography{sieve} 
\nocite{CP07,dantzig,Donoho,CRT1,NT,DM}
\end{document}